\documentclass[a4paper, 11pt]{article}
\usepackage[margin=1in]{geometry}
\usepackage{graphicx} % Required for inserting images
\usepackage{amsthm}
\usepackage{amsmath}
\usepackage{amssymb}
\usepackage{bbm}
\usepackage{pdfpages}
\usepackage{float}
\usepackage{comment}
\usepackage{natbib}
\usepackage{float}
\usepackage{booktabs}
\usepackage{url}
\usepackage{multirow}
\usepackage{hyperref}
\usepackage{cleveref}
\usepackage{authblk}

\newcommand{\indep}{\perp \!\!\! \perp}
\newcommand{\PP}{\mathbb{P}}
\newcommand{\QQ}{\mathbb{Q}}
\newcommand{\E}{\mathbb{E}}
\newcommand{\norm}[1]{\left\lVert #1 \right\rVert}

\newtheorem{theorem}{Theorem}

\newtheorem{lemma}{Lemma}
\newtheorem{assumption}{Assumption}
\newtheorem{remark}{Remark}

\title{Doubly Robust Machine Learning for Population Size Estimation with Missing Covariates: Application to Gaza Conflict Mortality}

\author{
\parbox{5cm}{\centering
Mateo Dulce Rubio \\
\small Center for Data Science \\
\small New York University \\
\small \texttt{mateo.d@nyu.edu}\\
\texttt{ }} 
\parbox{5cm}{\centering
Edward H. Kennedy \\
\small Department of Statistics \& Data Science \\
\small Carnegie Mellon University \\
\small \texttt{edward@stat.cmu.edu}}
\hspace{0.4cm}
\parbox{5.1cm}{\centering
Nicholas P. Jewell \\
\small Department of Medical Statistics \\
\small London School of Hygiene and Tropical Medicine \\
\small \texttt{nicholas.jewell@lshtm.ac.uk}}
}

\date{}

\begin{document}

\def\spacingset#1{\renewcommand{\baselinestretch}%
{#1}\small\normalsize} \spacingset{1}

\maketitle

\spacingset{1.5}
\begin{abstract}
    
    Population size estimation from capture-recapture data is central for studying hard-to-reach populations, incorporating auxiliary covariates to account for heterogeneous capture probabilities and recapture dependencies. However, missing attributes pose a critical methodological challenge due to reluctance to share sensitive information, data collection limitations, and imperfect record linkage. Existing approaches either ignore missingness or rely on \textit{a priori} imputation, potentially introducing substantial bias. In this work, we develop a novel nonparametric estimation framework using a Missing at Random assumption to identify capture probabilities under missing covariates. Using semiparametric efficiency theory, we construct one-step estimators that combine efficiency, robustness, and finite-sample validity: they approximately achieve the nonparametric efficiency bound, accommodate flexible machine learning methods through a doubly robust structure, and provide approximately valid inference for any sample size. Simulations demonstrate substantial improvements over naive imputation approaches, with our doubly robust ML estimators maintaining valid inference even at high missingness rates where competing methods fail. We apply our methodology to re-estimate mortality in the Gaza Strip from October 7, 2023, to June 30, 2024, using three-list capture-recapture data with missing demographic information. Our approach yields more conservative yet precise estimates compared to previous methods, indicating the true death toll exceeds official statistics by approximately 26\%. Our framework provides practitioners with principled tools for handling incomplete data in conflict settings and other applications with hard-to-reach populations.

\end{abstract}

\section{Introduction}

Population size estimation based on capture-recapture data has a long history in ecology and biology \citep{schwarz1999estimating, otis1978statistical}, and more recently has been applied to human populations to estimate the size of hard-to-reach groups such as people who use drugs \citep{jones2016problem}, living with HIV \citep{wesson2024evaluating}, homeless individuals \citep{coumans2017estimating}, victims of armed conflicts \citep{patrick2003many, jamaluddine2025lancet}, and other vulnerable communities. The foundational Lincoln-Petersen estimator, while elegant in its simplicity, relies on assumptions of homogeneous capture probabilities across all individuals and independent capture episodes \citep{lincoln1930calculating, petersen1896yearly}. However, these assumptions are rarely met in practice when working with human populations, where demographic or social factors introduce heterogeneity and dependence in the capture-recapture patterns, leading to biased and unreliable population size estimates \citep{manrique2021capture}.

Early proposals in the literature leveraged the fact that additional variables are typically available beyond the necessary capture-recapture profiles to incorporate heterogeneous capture probabilities into more general and realistic population size estimation frameworks \citep{chao1987estimating, huggins1989statistical, tilling1999capture, pollock2002use}. These covariates may include demographic characteristics, behavioral patterns, geographic information, or other factors that influence an individual's probability of being observed across different sources. Moreover, these techniques have also been recently extended to nonparametric models, moving away from unrealistic and restrictive parametric assumptions (e.g., linear relationships or complete stratification) and enabling more flexible modeling of complex capture-recapture processes \citep{das2024doubly, dulce2024population, you2021estimation}.

However, a key limitation hindering the full applicability and adoption of these methodologies is missing covariate information, which is prevalent when working with hard-to-reach populations \citep{ellard2015finding}. Individuals in these groups often have incomplete records due to various factors: reluctance to provide sensitive information, administrative challenges in data collection, high mobility, mistrust, or systematic avoidance of official services \citep{johnston2010sampling}. Moreover, probabilistic record linkage across capture-recapture episodes, a necessary first step in population size estimation when unique identifiers are unavailable, frequently generates missing covariate information in the linked data due to conflicting values between matched records \citep{fellegi1969theory, binette2022almost}. 

This missing data problem creates a fundamental tension in population size estimation: incorporating covariate information is essential for accounting for capture heterogeneity and recapture dependencies, yet naive approaches to handling missingness such as complete case analysis or simple imputation can introduce significant bias in downstream analysis \citep{zwane2007analysing, little2019statistical}, particularly when missingness is related to both the covariates and the capture process itself. This paper addresses this critical gap by developing a novel nonparametric estimation framework for population size estimation in the presence of missing covariates. Our approach combines a Missing at Random assumption with doubly robust machine learning methods to provide flexible estimation that is robust to misspecification of either the missingness mechanism or the heterogeneous capture probability models. This provides practitioners with principled statistical tools for handling missing covariates when estimating the size of hard-to-reach populations.

\paragraph*{Paper Organization and Contributions.} After summarizing the relevant literature on population size estimation in Section \ref{sec:literature}, the key contributions of this work are:

\begin{itemize}
    \item We establish a novel identification expression for the capture probability under a Missing at Random assumption combined with a no highest-order interaction assumption in \Cref{sec:identification}.

    \item In \Cref{sec:estimation}, we develop nonparametric doubly robust estimators for the capture probability and total population size that are approximately optimal and normally distributed for any finite sample size, even when flexible machine learning methods are used for nuisance function estimation.

    \item We empirically demonstrate the robustness and efficiency gains of our approach in finite samples through simulation studies in Section \ref{sec:experiments}, comparing against naive imputation methods and plug-in estimators typically used in capture-recapture applications.

    \item We apply our methodology to estimate conflict mortality in the Gaza Strip from October 7, 2023, to June 30, 2024, obtaining 59,441 estimated deaths (95\% CI: 50,708-68,173), approximately 10\% lower than previous estimates. Yet, our results imply that conflict deaths exceed official statistics by approximately 26\%, highlighting the practical importance of robust statistical methods for conflict casualty estimation with incomplete data (Section \ref{sec:gaza}).
    
\end{itemize}

\section{Related Work on Population Size Estimation}
\label{sec:literature}

In this section we summarize the relevant literature on capture-recapture methods for population size estimation, with particular focus on approaches that incorporate covariates and handle missing data. We begin with classical methods, discuss extensions that account for heterogeneity and list dependence, and conclude with recent developments in nonparametric and machine-learning-based approaches.

\subsection{Classic Capture-Recapture Methods}

Capture-recapture studies, or multiple systems estimation, aim to estimate total population size by combining multiple partial samples from the same closed population and tracking which individuals are observed across different capture occasions. We refer to each capture episode as a \textit{list}, where individuals are recorded and can be linked (albeit imperfectly) across different lists. The key challenge lies in the unidentifiability of unobserved individuals: we only observe those who are captured at least once, yet the target population includes those who remain completely undetected across all available lists. This creates the fundamental problem in capture-recapture settings of extrapolating from the observed capture patterns to the unobserved population to estimate the total population size \citep{fienberg1972multiple, bird2018multiple}.

The simplest case involves two capture episodes. Consider a population of unknown size $n$, where $N_1$ individuals are captured in the first list, $N_2$ in the second one, and $N_{12}$ individuals appear in both samples. The total number of unique observed individuals is then $N = N_1 + N_2 - N_{12}$. We aim to estimate the total population size $n = N + N_0$ using the observed counts and capture patterns, where $N_0$ is the population unobserved by both lists.

Assume that (1) the population is closed between sampling occasions, (2) all individuals have equal capture probabilities bounded away from zero, (3) capture events are independent, and (4) no identification errors occur across lists. Moreover, let $p_1$ and $p_2$ denote the capture probabilities for lists 1 and 2. Therefore, from our assumptions it follows that the expected observed overlap between lists is 
$$E[N_{12}] = np_1p_2. $$ 
Since the marginal capture probabilities can be estimated as $\widehat{p}_1 = N_1/n$ and $\widehat{p}_2 = N_2/n$, by substituting into the overlap equation we can derive the \textit{Lincoln-Petersen estimator} for the total population size $\widehat{n}_{LP}$ \citep{lincoln1930calculating, petersen1896yearly}:
\begin{equation}
    N_{12} = n \frac{N_1}{n} \frac{N_2}{n} \Longrightarrow \widehat{n}_{LP} = \frac{N_1N_2}{N_{12}}.
\end{equation}

Multiple extensions to this setting have been proposed to relax the independence and homogeneity assumptions. When multiple lists are available, log-linear models use a weaker notion of independence allowing lower-order interactions between capture occasions, but assuming that there is no interaction effect involving all capture events simultaneously \citep{fienberg1972multiple, bishop2007discrete}. To accommodate individual heterogeneity and list dependence, \cite{otis1978statistical} introduced the $M_h$, $M_t$, and $M_{th}$ models that allow for individual-specific or time-varying capture probabilities while preserving independence of capture events. However, these classical approaches estimate heterogeneous capture probabilities from observed capture patterns alone, without incorporating auxiliary covariate information.

\subsection{Covariate-Adjusted Methods}

Typically, capture-recapture studies collect additional variables during capture events, such as demographic characteristics (age, gender), behavioral patterns (substance use, service utilization), geographic information (residence location, mobility patterns), or socioeconomic factors (employment status, income) \citep{pollock2002use}. It is natural then to incorporate this information with two aims: accounting for heterogeneous capture probabilities across individuals, and relaxing identification assumptions to hold conditionally given observed covariates \citep{aleshin2024central}. In general, including covariate information in capture-recapture models enables more flexible and realistic population size estimation.

The most natural covariate-adjusted approach is stratification, where the population is divided into homogeneous subgroups based on discrete covariates and separate analyses are conducted within each stratum \citep{bishop2007discrete, kurtz2014local}. This approach accommodates heterogeneity across covariate levels and assumes conditional independence, where captures are independent within each stratum. However, stratification suffers from the curse of dimensionality with multiple covariates and can lead to sparse counts in high-dimensional settings.

More flexible approaches model capture probabilities directly as functions of individual covariates through regression-based methods \citep{huggins1989statistical, pollock2002use}. These approaches accommodate both continuous and discrete covariates through generalized linear models, enabling more nuanced modeling of capture heterogeneity while maintaining interpretability \citep{tilling1999capture}. However, parametric specifications may still impose restrictive functional form assumptions on the relationship between covariates and capture probabilities, which can introduce bias when these assumptions are violated.

\subsection{Nonparametric Approaches}

Recent developments have moved away from parametric specifications to accommodate more flexible relationships between covariates and capture probabilities. These approaches use nonparametric methods such as kernel smoothing, local polynomials, or Bayesian latent-class models to account for capture heterogeneity without imposing restrictive functional forms \citep{chen2000nonparametric, hwang2011semiparametric, manrique2016bayesian}. These methods offer greater flexibility than their parametric counterparts, but they typically rely on plug-in estimators that substitute estimated nuisance functions directly into total population size estimation \citep{huggins2007non}. This plug-in strategy, while intuitive, can suffer from slow rates of convergence and potentially invalid inference when combined with flexible machine learning estimation of nuisance capture probabilities \citep{vanderlaan2003unified, ehk_semiparametric}.

Closer to our work are the proposals by \cite{das2024doubly} and \cite{dulce2024population}, who leverage semiparametric efficiency theory to develop doubly robust machine learning estimators that approximately achieve the efficiency bound while accommodating flexible methods for nuisance function estimation. Specifically, \cite{das2024doubly} introduced doubly robust population size estimation in two-list capture-recapture settings, while \cite{dulce2024population} extended this framework to multiple lists under a log-linear model with no highest-order interaction. However, these methods assume complete covariate information for all captured individuals, an assumption rarely satisfied in practice when working with observational data on hard-to-reach populations where missing data is prevalent.

\subsection{Dealing with Missing Data}

Despite these methodological advances, missing covariate information remains an endemic challenge in capture-recapture studies, due to administrative challenges, mistrust in data collectors, or imperfect record linkage \citep{johnston2010sampling, ellard2015finding, binette2022almost}. Applied studies typically address this issue through a priori imputation methods (e.g., \citep{white2011multiple}), filling in missing values before conducting population size estimation \citep{patrick2003many, jamaluddine2025lancet}. However, these approaches can introduce substantial bias, particularly when missingness is related to both the covariates and the capture process \citep{zwane2007analysing, manrique2021capture}. More principled methods remain sparse and limited: \citet{manrique2019estimating} proposed a Bayesian framework for stratified capture-recapture with incomplete stratification in a single discrete covariate, but this approach does not extend readily to multiple or continuous covariates.

In this work we overcome these critical limitations by developing a nonparametric doubly robust estimation framework that directly accommodates missing covariate information with a Missing at Random assumption. Using semiparametric efficiency theory, we construct one-step machine learning estimators that combine efficiency, robustness, and finite-sample valid inference for any sample size. Unlike prior work, our framework handles multiple (potentially high-dimensional or continuous) covariates and provides formal finite-sample guarantees even when missingness rates are substantial.
\section{Preliminaries}
\label{sec:identification}

In the following we formally define the problem setup and notation for population size estimation of heterogeneous populations with missing covariates. We then establish an identification result for the capture probability and total population size under a standard no highest-order interaction in a log-linear model, combined with a Missing at Random assumption conditional on observables.

\subsection{Problem Formulation}

We are interested in estimating the size $n$ of a closed population using data from $K$ capture-recapture lists. We assume that each uniquely identified individual $i \in \{1, \dots,n \}$ in the target population is associated with four variables:

\begin{itemize}
    \item A binary \textit{capture profile} $Y_i = (Y_{i1}, \dots, Y_{iK}) \in \{0,1\}^K$, where $Y_{ik} = 1$ indicates presence on list $k$.
    
    \item A vector of \textit{always-observed} covariates $V_i \in \mathbb{R}^{d_V}$.
    
    \item A vector of \textit{potentially missing covariates} $X_i \in \mathbb{R}^{d_X}$.
    
    \item A binary \textit{missing indicator} $R_i \in \{0,1\}$, where $R_{i} = 1$ if all covariates $X_i$ were observed.
\end{itemize}

The observed data is fundamentally limited in two ways. First, we only have information for individuals captured on at least one list, meaning we observe only the $N$ units satisfying $Y_i \neq 0$ (where $0$ denotes the $K$-dimensional zero vector), with $N \leq n$. Second, even among captured individuals, the $X$ covariates may be unobserved  according to the missingness indicator $R_i$. Our observed dataset is therefore $\mathcal{D} = \{Z_i = (Y_i, V_i, R_iX_i, R_i) : Y_i \neq 0\}_{i = 1}^N$, where $R_i X_i$ equals the observed value for $X_i$ when $R_i = 1$ and is treated as missing when $R_i = 0$.

We assume that the covariate and conditional list membership distributions are identical across individuals and independent of other units. This implies that the random vectors $Z_i = (Y_i, V_i, R_iX_i, R_i)$ are \textit{iid} with respect to an unknown distribution $\mathbb{P}$. The number of captured individuals $N = \sum_{i=1}^n \mathbbm{1}(Y_i \neq 0)$ is thus a binomial random variable with $N \sim \text{Bin}(n, \psi)$, where $\psi = \mathbb{P}(Y \neq 0)$ represents the marginal \textit{capture probability} that an individual appears on at least one list. Since $\mathbb{E}[N] = n\psi$, the population size can be estimated via $\widehat{n} = N/\widehat{\psi}$, effectively reducing our problem to estimating the inverse capture probability $\psi^{-1}$.

Crucially, since individuals with $Y_i = 0$ are by definition not included in our sample, the observed capture-recapture data arise from the selection-biased distribution \citep{qin2017biased}
\begin{equation}
    \mathbb{Q}(Y, V, X, R) = \mathbb{P}(Y, V, X, R \mid Y \neq 0),
\end{equation}
governing the underlying joint distribution conditional on being observed. To connect this observable distribution $\mathbb{Q}$ with the target distribution $\mathbb{P}$, we define the following full conditional \textit{q-probabilities} for any capture profile $y \neq 0$:
\begin{equation}
    q_y(v, x) = \mathbb{Q}(Y = y \mid V = v, X = x),
\end{equation}
as the conditional distribution of capture profiles given covariates among the captured population. Additionally, let $\gamma(v, x) = \mathbb{P}(Y \neq 0 \mid V = v, X = x)$ denote the \textit{conditional capture probability}. These quantities are related for all $y \neq 0$ through \citep{das2024doubly}
\begin{equation}
\begin{split}
\label{eq:p_qgamma}
\mathbb{P}(Y = y \mid V = v, X = x) & =  \PP(Y=y \mid V=v, X=x, Y\neq 0)\PP(Y\neq 0 \mid V=v, X=x) \\
& = q_y(v, x)  \gamma(v, x).
\end{split}
\end{equation}

Furthermore, it is known that the inverse marginal capture probability $\psi^{-1}$ can be written as the harmonic mean of the conditional capture probability $\gamma(v,x)$:
\begin{equation}
\psi^{-1} = \mathbb{E}_{\mathbb{Q}}\left[\frac{1}{\gamma(V, X)}\right] = \int \frac{1}{\gamma(v, x)} \mathbb{Q}(V = v, X = x)dvdx,
\label{eq:harmonic_mean}
\end{equation}
where the expectation is taken with respect to the observable covariate distribution $\mathbb{Q}(V, X) = \mathbb{P}(V, X \mid Y \neq 0)$. This identity shows that if we can identify and estimate $\gamma(V, X)$ from the observed data, we can recover $\psi^{-1}$ and consequently estimate the total population size $n$. 

However, since we aim to estimate $\PP(Y \neq 0)$ by extrapolating  information from the observable distribution $\QQ$ to the target distribution $\PP$, it is necessary to introduce additional \textit{identification assumptions} to enable such extrapolation from the observed to the unobserved populations \citep{aleshin2024central}. Moreover, the presence of missing covariates $X$ further complicates both the identification and estimation of $\gamma(V, X)$ and the evaluation of the harmonic mean in equation \eqref{eq:harmonic_mean}.

\subsection{Identification of the Capture Probability}

To achieve identification of the inverse capture probability $\psi^{-1}$, we introduce two key assumptions that address the challenges of biased sampling (we only observe $\QQ$ rather than $\PP$) and missing covariate information. These identification assumptions are necessary so that the observable data distribution $\QQ$ is informative about the target estimand $\psi = \PP(Y \neq 0)$. 

Formally, consider a conditional log-linear model over the $K$ available lists, where the log-conditional probability for each capture profile $y=(y_1, \dots, y_K) \in \{0,1\}^{K}$ is modeled as: 
\begin{equation}
\label{eq:log_lineal}
    \log\left[\PP\left(Y = y \mid V = v, X=x\right)\right] = \alpha_0(v, x) + \sum_{y' \neq 0} \alpha_{y'}(v, x) \prod_{j:~ y'_j = 1}y_j,
\end{equation}
where the coefficients $\alpha_y(v,x)$ vary arbitrarily with respect to the covariates $(V=v, X=x)$ to accommodate heterogeneous capture probabilities across individuals and dependencies between lists  \citep{bishop2007discrete, fienberg1972multiple, dulce2024population}. 

Note that log-linear models are unidentifiable since there are $2^K$ parameters $\alpha_y$ to be estimated from only $2^K - 1$ observable capture profiles (all profiles but $y=0$). Therefore, the standard approach to achieve identification is to assume \textit{no highest-order interaction} (NHOI): $\alpha_{1}(v,x) = 0$ for all $(v,x)$, where 1 denotes the always-captured pattern $y = (1,\dots, 1)$. This assumption states that the probability of being observed by all lists, $\PP(Y = (1,\dots, 1) \mid V,X)$, can be expressed using only lower-order interactions between lists, and rules out a $K$-way interaction effect involving all capture events simultaneously. 

\begin{assumption}
\label{ass:nhoi}\textbf{(No Highest-Order Interaction - NHOI)}: 
The K-way interaction term in the conditional log-linear model \eqref{eq:log_lineal} satisfies $\alpha_{1}(v,x) = 0$ for all covariate values $(v,x)$.
\end{assumption}

The NHOI assumption is standard in the capture-recapture literature and equivalent to a (conditional) independence assumption in two-lists settings. In contrast, it is strictly weaker than an independence assumption between capture events when $K \geq 3$ lists are available. For a detailed discussion on the interpretation of the NHOI assumption for $K$ lists and comparison to alternative identification strategies, see \cite{dulce2024population}.

\Cref{ass:nhoi} yields the following closed-form expression for the inverse conditional capture probability \citep{you2021estimation}:
\begin{equation}
\label{eq:gamma_nhoi}
    \frac{1}{\gamma(v,x)} = 1 + \exp\left(\sum_{y \neq 0} (-1)^{|y| + 1}\log(q_y(v, x)) \right),
\end{equation}
where $|y| = \sum_{k=1}^{K} y_k$ denotes the number of lists present in capture profile $y = (y_1, \dots, y_K)$ (i.e., the $L^{1}-$norm of $y$). However, the full conditional $q$-probabilities $q_y(v, x)$ remain unidentified from the observed data since we do not always observe the covariates $X$. To address this identification challenge, we introduce a \textit{Missing at Random} assumption:

\begin{assumption}
\label{ass:mar}\textbf{(Missing at Random - MAR)}: The missingness mechanism $R$ is conditionally independent of the missing covariates $X$, given the always-observed covariates $V$ and capture profile $Y$:
$$X \indep R \mid V, Y.$$
\end{assumption}

Intuitively, \Cref{ass:mar} states that whether $X$ is observed or not does not depend on the unobserved values of $X$ itself, after conditioning on the observable information $(V, Y)$. This is a standard assumption in the missing data literature and is weaker than assuming missingness is completely at random \citep{little2019statistical}. Moreover, it is plausible in many capture-recapture settings where missingness patterns depend on observable characteristics and capture behavior but not on the specific unobserved covariate values.

Under the Missing at Random assumption, we can then identify the full conditional $q$-probabilities for all $y\neq 0$ as follows:
\begin{align*}
    \QQ(Y=y \mid V = v, X=x) & = \frac{\QQ(X=x \mid Y=y, V=v)\QQ(Y=y \mid V=v)}{\QQ(X=x \mid V=v)} \\
    & \overset{MAR}{=} \frac{\QQ(X=x \mid Y=y, V=v, R=1)\QQ(Y=y \mid V=v)}{\sum_{y\neq0}\QQ(X=x \mid Y=y, V=v, R=1)\QQ(Y=y \mid V=v)} \\
    & = \frac{\lambda_x(y,v)q_y(v)}{\sum_{y\neq 0}\lambda_x(y,v)q_y(v)},
\end{align*}
where $\lambda_x(y,v) = \QQ(X=x \mid Y=y, V=v, R=1)$ is the conditional covariate distribution of 
$X$ among fully-observed individuals with $R=1$, and $q_y(v) = \QQ(Y=y \mid V=v)$ are the $q$-probabilities depending only on the always-observed covariates $V$. Notably,  both $\lambda_x(y,v)$ and $q_y(v)$ are identifiable from the observed data: the former from individuals with $R=1$ and the latter from all captured individuals.

Substituting this expression for the full conditional $q$-probabilities into equations \eqref{eq:gamma_nhoi} and \eqref{eq:harmonic_mean} we obtain our main identification result. 

\begin{lemma}
\label{prop:identification} \textbf{(Identification Result)} 
Assume no highest-order interaction (\Cref{ass:nhoi}) in a conditional log-linear model with Missing at Random covariates (\Cref{ass:mar}). Under positivity conditions, the inverse capture probability $\psi^{-1}$ is identified by
\begin{align}
\label{eq:psi_identified}
    \frac{1}{\psi} & = \E_{\QQ}\left[\frac{1}{\gamma(V,X)}\right]\\ 
    & = \E_{\QQ}\left[1 + \exp\left(\sum_{y \neq 0} (-1)^{|y| + 1}\left\{\log(\lambda_X(y,V)) + \log(q_y(V))\right\} - \log\left(\sum_{y\neq 0}\lambda_X(y,V)q_y(V)\right)\right)\right], \nonumber
\end{align}
where $q_y(v) = \QQ(Y=y \mid V=v)$ and $\lambda_x(y,v) = \QQ(X=x \mid Y=y, V=v, R=1)$ are both identifiable from the observed data $\mathcal{D}$.
\end{lemma}

\Cref{prop:identification} establishes that the population size is identifiable under our two key assumptions: MAR to handle missing covariates and NHOI to extrapolate to the unobserved population. Crucially, our identification expression \eqref{eq:psi_identified} depends only on the \textit{nuisance functions} $q_y(v)$ and $\lambda_x(y,v)$, which can be estimated directly from the observed data $\mathcal{D}$ using any estimation procedure. In the following section, we develop flexible yet robust nonparametric estimators for the inverse capture probability $\widehat{\psi}^{-1}$, and the corresponding total population size $\widehat{n} = N/\widehat{\psi}$.

\begin{remark}
    Our identification strategy leverages a conditional log-linear model and NHOI to assume certain independence structure among capture events in the population distribution $\PP$. Importantly, these assumptions impose no restrictions on the \textit{observable} distribution $\QQ$, preserving the nonparametric nature of our model with respect to the nuisance functions that identify $\psi$. This approach enables us to develop efficient nonparametric estimators while avoiding additional restrictive parametric assumptions \citep{dulce2024population}.
\end{remark}

\subsection{Notation}
We summarize the main nuisance functions used to identify the capture probability $\psi$:
\begin{align*}
    \gamma(v,x) & = \PP(Y\neq 0 \mid V=v, X=x) && \text{(conditional capture probability)} \\
    q_y(v, x) & = \mathbb{Q}(Y=y \mid V=v, X=x) && \text{(full conditional $q$-probabilities)} \\
    q_y(v) & = \mathbb{Q}(Y=y \mid V=v) && \text{($V$-only conditional $q$-probabilities)} \\
    \lambda_x(y,v) & = \mathbb{Q}(X=x \mid Y=y, V=v, R=1) && \text{(covariate distribution among $R=1$)} \\
    \pi(y, v) & = \mathbb{Q}(R=1 \mid Y=y, V=v) && \text{(missingness propensity score)} 
    %q(x \mid v) & = \QQ(X=x \mid V=v) = \sum_{y\neq 0} \lambda_x(y,v)q_y(v) && \text{(conditional covariate distribution)} \\
\end{align*}

\begin{remark}
    We assume positivity conditions for the nuisance functions $q_y(v) > 0$, $\lambda_x(y,v) > 0$, and $\pi(y,v) > 0$, with probability one, for all combinations in the support, ensuring well-defined inverse probability weights and identifiability.
\end{remark}

\section{Nonparametric Estimation}
\label{sec:estimation}

Having established identification of the inverse capture probability $\psi^{-1}$ under the NHOI and MAR assumptions, we now develop estimation procedures for the target parameter. Our approach leverages semiparametric efficiency theory to construct machine learning (ML) based estimators that approximately achieve the efficiency bound while being robust to misspecification of either the missingness mechanism or the heterogeneous $q$-probabilities. This ``doubly robust'' property manifests through a second-order error term depending on products of nuisance function estimation errors. We begin by deriving the efficient influence function (EIF) for $\psi^{-1}$, which characterizes the optimal asymptotic MSE and guides our estimator construction. We then propose two estimation strategies: a plug-in estimator and a doubly robust one-step estimator that incorporates bias correction through the EIF. Finally, we develop the corresponding estimators for the population size $n$ and derive approximately valid confidence intervals with finite-sample guarantees. Throughout, we maintain our nonparametric model allowing for flexible estimation of the nuisance functions using modern machine learning methods. 

\subsection{Influence Function and Efficiency Bound}

We start by deriving the Efficient Influence Function (EIF) for the inverse capture probability $\psi^{-1}$ from the identification expression in \Cref{prop:identification}. The EIF is critical in our setting as it: (i) characterizes the lowest achievable tsquared error (MSE) in our nonparametric model, (ii) provides a recipe to construct efficient estimators, and (iii) clarifies the regularity conditions needed for efficiency \citep{bickel1993efficient, semiparametric, ehk_semiparametric}. These properties are essential when the true functional forms of the nuisance functions are unknown, enabling data-adaptive estimation of heterogeneous capture probabilities and missingness patterns while maintaining robustness and statistical efficiency. 

The following lemma establishes the EIF for our parameter of interest $\psi^{-1}$ in an unrestricted nonparametric model. All proofs are presented in the Appendix.

\begin{lemma}
\label{lemma:if_psi}
Assume (i) no highest-order interaction in a conditional log-linear model with covariates $(V,X)$ (\Cref{ass:nhoi}), (ii) missing at random (\Cref{ass:mar}), and (iii) positivity. Then, the Efficient Influence Function for $\psi^{-1} = \mathbb{E}_{\mathbb{Q}}[\gamma(V, X)^{-1}]$ is given by:
\begin{align}
\label{eq:if_psi}
    \phi(Z) =  \sum_{y\neq 0} (-1)^{1+|y|} & \left\{ \frac{R\mathbbm{1}(Y=y)}{\pi(y,v)}\left(\frac{1}{q_y(v,x)}\left(\frac{1}{\gamma(v,x)}-1\right) - \frac{1}{q_y(v)}\mathbb{E}\left[\frac{1}{\gamma(V,X)} - 1 ~\Bigg|~ V\right]\right) \right. \nonumber \\ 
    & \left. +\left(\frac{\mathbbm{1}(Y=y)}{q_y(v)}\right)\mathbb{E}\left[\frac{1}{\gamma(V, X)} - 1 ~\Bigg|~ V\right]  \right\}  - \left(\frac{1}{\psi} - 1\right),
\end{align}
where the nuisance functions $q_y(v), q_y(v,x), \pi(y,v)$ and $\gamma(x,v)$ are defined in the notation subsection above.
\end{lemma}

The nonparametric efficiency bound can then be computed as the variance of the EIF $\phi(Z)$ for $\psi^{-1}$ from \Cref{lemma:if_psi}. This bound gives the best possible asymptotic MSE achievable by \textit{any} estimator for the capture probability in a local minimax sense \citep{semiparametric}. That is, once an estimator is shown to achieve this bound no further improvement is possible in large samples, without imposing additional parametric restrictions on the nuisance functions. However, asymptotic results may be insufficient in capture-recapture settings with moderate sample sizes, so we show in the next section that our proposed one-step ML estimator based on the EIF also enjoys approximately optimal efficiency guarantees for any sample size.

%\begin{theorem}
%\label{thm:eff_bound}
%Under the assumptions of Lemma \ref{lemma:if_psi}, the nonparametric efficiency bound for $\psi^{-1} = \mathbb{E}_{\mathbb{Q}}[\gamma(V,X)^{-1}]$ is:
%\begin{equation}
%\label{eq:eff_bound}
%    \text{Var}(\phi(Z)) = 
%\end{equation}
%\end{theorem}

\subsection{Estimation Strategies}

In the following we propose two estimation approaches for the inverse capture probability $\psi^{-1}$. First, we present a baseline plug-in estimator that, while intuitive, typically exhibits first-order bias when combined with flexible nuisance function estimation. To address this limitation, our primary contribution is a doubly robust ML estimator that approximately achieves the efficiency bound and is approximately normally distributed for any finite sample size.

\subsubsection{Plug-in Estimator}

A natural baseline approach, commonly known as a \textit{plug-in estimator}, directly substitutes estimated nuisance functions into the identified expression for $\psi^{-1}$ from \Cref{prop:identification} and approximates the expectation with an empirical average \citep{fienberg1972multiple, dulce2024population}. However, since $X$ is not always observed, we cannot directly compute the empirical average over $(V,X)$ pairs. Instead, we use the law of iterated expectations to write 
$$\frac{1}{\psi} = \E_{\QQ(V,X)}\left[\frac{1}{\gamma(V,X)}\right] = \E_{\QQ(V)}\left[\E_{\QQ(X | V)}\left[\frac{1}{\gamma(V,X)} ~\Bigg|~ V \right]\right].$$

The plug-in estimator $\widehat{\psi}^{-1}_{\text{pi}}$ is thus constructed as follows: (i) estimate nuisance functions $\widehat{\lambda}_x(y,v)$ and $\widehat{q}_y(v)$ using any parametric or nonparametric ML method; (ii) combine these to construct conditional capture probabilities $\widehat{\gamma}(v,x)$ for each unit as specified in \Cref{prop:identification}; (iii) for each observed $v_i$, estimate the inner conditional expectation by integrating $\widehat{\gamma}^{-1}(v_i, x)$ against the estimated conditional covariate distribution $\widehat{q}(x | v_i)$, and (iv) compute the outer expectation using the empirical average over the observed units:
\begin{equation}
\label{eq:plugin_psi}
\widehat{\frac{1}{\psi_{\text{pi}}}} = \frac{1}{N}\sum_{i=1}^N \sum_{x} \frac{1}{\widehat{\gamma}(v_i, x)} \widehat{q}(x | v_i),
\end{equation}
where $\widehat{q}(x | v) = \sum_{y\neq 0} \widehat{\lambda}_x(y,v)\widehat{q}_y(v)$ under \Cref{ass:mar}.

%The plug-in estimator is computationally straightforward and only requires estimation of the nuisance functions $q_y(v)$ and $\lambda_x(y,v)$ to construct $\gamma(v,x)$ and take its harmonic mean. 
While this plug-in strategy is conceptually straightforward, it faces two critical limitations when nuisance functions are estimated using flexible machine learning methods. First, using these data-adaptive estimators can lead to slower than $\sqrt{N}$ convergence, resulting in statistical inefficiency. Second, the limiting distribution becomes unknown, prohibiting standard inference \citep{chernozhukov2018double, vanderlaan2003unified, ehk_semiparametric}. These deficiencies motivate the development of our bias-corrected one-step estimator.

\subsubsection{One-step ML Estimator}

We now leverage semiparametric efficiency theory to overcome the limitations of the plug-in estimator and achieve both efficiency and robustness in our nonparametric model. The key insight is to construct a one-step ML estimator by augmenting the plug-in estimator with a bias-correction term derived from the EIF in \Cref{lemma:if_psi}. This debiasing approach, well-established in the semiparametric literature \citep{semiparametric, bickel1993efficient}, eliminates first-order bias and yields an estimator that can achieve the efficiency bound even when nuisance functions are estimated using flexible machine learning methods. 

The one-step ML estimator for the inverse capture probability $\widehat{\psi}^{-1}_{os}$ takes the form:
$$\widehat{\frac{1}{\psi_{\text{os}}}} = \widehat{\frac{1}{\psi_{\text{pi}}}} + \frac{1}{N}\sum_{i=1}^N \widehat{\phi}(Z_i),$$
where $\widehat{\phi}(Z_i)$ is the estimated EIF from \eqref{eq:if_psi} with nuisance functions replaced by their estimates. Note that, analogous to the plug-in estimator, any parametric or nonparametric method can be employed to estimate nuisance functions $\widehat{\lambda}_x(y,v)$ and $\widehat{q}_y(v)$. Expanding this expression yields:
\begin{align}
\label{eq:dr_psi}
     \widehat{\frac{1}{\psi_{os}}} =  & \frac{1}{N}\sum_{i=1}^N  \sum_{y\neq 0} (-1)^{1+|y|}\left\{ \left(\frac{\mathbbm{1}(Y_i=y)}{\widehat{q}_y(v_i)}\right)\widehat{\mathbb{E}}\left[\frac{1}{\widehat{\gamma}(V, X)} - 1 ~\Bigg|~ V=v_i\right]  \right. \\ \nonumber
    & \left. + \frac{\mathbbm{1}(Y_i=y)R_i}{\widehat{\pi}(y,v_i)}\left(\frac{1}{\widehat{q}_y(v_i, x_i)}\left(\frac{1}{\widehat{\gamma}(v_i, x_i)}-1\right) - \frac{1}{\widehat{q}_y(v_i)}\widehat{\mathbb{E}}\left[\frac{1}{\widehat{\gamma}(V,X)} - 1 ~\Bigg|~ V=v_i\right]\right) \right \} + 1,
\end{align}
where $\widehat{q}_y(v, x) = \widehat{\lambda}_x(y,v)\widehat{q}_y(v)/\sum_{y'\neq 0}\widehat{\lambda}_x(y',v)\widehat{q}_{y'}(v)$ is the plug-in estimator under MAR for the full conditional $q$-probabilities $q_y(v,x)$. 

Intuitively, the first term in our one-step ML estimator provides an estimate of $\psi^{-1}$ using only the always-observed covariates $V$ to model heterogeneous capture probabilities $q_y(v)$. The second term corrects for this approximation by incorporating information from the potentially missing covariates $X$ through (augmented) inverse probability weighting. Specifically, observations where $X$ is observed ($R_i=1$) are reweighted by $\widehat{\pi}(y_i,v_i)^{-1}$ to represent the full covariate distribution, accounting for residuals between using the full conditional capture probabilities $\widehat{q}_y(v,x)$ and $\widehat{\gamma}(v,x)^{-1}$ versus their $V$-only counterparts $\widehat{q}_y(v)$ and $\widehat{\E}[\widehat{\gamma}(V,X)^{-1}|V]$.

We now establish key properties of the one-step ML estimator that hold for any finite sample size $N$, even when nuisance functions are estimated using flexible machine learning methods. Specifically, we demonstrate that the estimator is (i) nearly optimal in the sense that it closely approximates the empirical average of the EIF, (ii) ``doubly robust'' with respect to nuisance function estimation, and (iii) approximately normally distributed enabling valid inference. These finite-sample guarantees are particularly important in capture-recapture applications, where the sample size $N$ is bounded by the fixed population size $n \geq N$
rather than growing to infinity. 

\begin{theorem} 
\label{thm:optimal} 
For any sample size $N$, the proposed one-step ML estimator $\widehat{\psi}^{-1}_{os}$ approximates the empirical average of the EIF, $\QQ_N(\phi)$, centered at the true inverse capture probability. Specifically, for any predetermined error tolerance $\eta > 0$:
\begin{equation}
    \PP\left(\Big|(\widehat{\psi}^{-1}_{os} - \psi^{-1}) - \QQ_N(\phi)\Big| \leq \eta\right) \geq 1 - \left(\frac{1}{\eta^2}\right) \E_{\QQ}\left[\widehat{R}_2^2 + \frac{1}{N}\norm{\widehat{\phi} - \phi}^2\right],
\end{equation}
where $\phi$ is the EIF from \eqref{eq:if_psi} and $\widehat{R}_2$ is the reminder term in the von Misses expansion given by
\begin{align}
    \widehat{R}_2 = \E_{\QQ}&\left[  \left(\frac{1}{\widehat{\gamma}(V,X)}-1\right)\sum_{y\neq 0} (-1)^{1+|y|}\left(\frac{\pi(y,V) - \widehat{\pi}(y,V)}{\widehat{\pi}(y, V)}\right)\left(\frac{q_y(V,X) - \widehat{q}_y(V,X)}{\widehat{q}_y(V,X)} - \frac{q_y(V) - \widehat{q}_y(V)}{\widehat{q}_y(V)}\right) \nonumber \right.\\
    & + \left(\frac{1}{\widetilde{\gamma}(V,X)}-1\right)\sum_{y_i\neq y_j \neq 0} (-1)^{|y_i| + |y_j|}\frac{(q_{y_i}(V,X) - \widehat{q}_{y_i}(V,X))}{\widetilde{q}_{y_i}(V,X)}\frac{(q_{y_j}(V,X) - \widehat{q}_{y_j}(V,X))}{\widetilde{q}_{y_j}(V,X)} \nonumber \\ 
    & + \left. \left(\frac{1}{\widetilde{\gamma}(V,X)}-1\right)\sum_{|y|\neq 0, ~even} \left(\frac{q_{y}(V,X) - \widehat{q}_{y}(V,X)}{\widetilde{q}_{y}(V,X)}\right)^2\right],
    \label{eq:r2_bound}
\end{align}
for some intermediate $\widetilde{q}(V,X)$ between $q(V,X)$ and $\widehat{q}(V,X)$.
\end{theorem}

Crucially, the remainder term $\widehat{R}_2$ is of second order and exhibits a \textit{doubly robust} structure through products of nuisance function estimation errors. The first component involves products of errors in the missingness propensity score $\pi(y,v)$ with differences between $V$-only and full conditional $q$-probabilities. The second and third components capture cross-products of errors in $q_y(v,x)$ across different capture profiles, resembling the second order reminder for the one-step estimator in log-linear models with complete covariate information derived in \cite{dulce2024population}. These terms depend on the estimation error of the full conditional capture probabilities via $\widehat{q}_y(v, x) = \widehat{\lambda}_x(y,v)\widehat{q}_y(v)/\sum_{y'\neq 0}\widehat{\lambda}_x(y',v)\widehat{q}_{y'}(v)$, which itself depends on estimating the nuisance functions $q_y(v)$ and $\lambda_x(y,v)$ accurately. Notably, while estimating the conditional density $\lambda_x(y,v)$ can be challenging for continuous or high-dimensional $X$ \citep{kennedy2023density}, estimation errors in $\widehat{\lambda}_x(y,v)$ affect $\widehat{R}_2$ only through their impact on $\widehat{q}_y(v,x)$. This averaging can provide some robustness, as localized estimation errors in the density may partially cancel when aggregated into the capture probability estimates. Finally, the error term in estimating the EIF, $\frac{1}{N}||\widehat{\phi} - \phi||^2$, also depends on these nuisance functions but contributes at a lower order due to the $1/N$ factor.

The doubly robust property of the proposed one-step ML estimator is particularly valuable for finite-sample inference as \Cref{thm:optimal} holds for any sample size $N$, not just asymptotically. That is, the finite-sample deviation $|(\widehat{\psi}^{-1}_{os} - \psi^{-1}) - \mathbb{Q}_N(\phi)|$ is controlled by the second-order remainder $\widehat{R}_2$ for any fixed $N$, providing approximate valid and optimal inference even in moderate sample sizes typical of capture-recapture studies. Moreover, this results provide asymptotic guarantees even when nuisance functions are estimated using flexible machine learning methods, such that the one-step estimator achieves the efficiency bound when $\mathbb{E}[\widehat{R}^2_2] = o_P(N^{-1})$. For instance, standard machine learning rates of $o_P(N^{-1/4})$ for each nuisance function suffice, as the product structure in $R_2$ squares these individual convergence rates.

\begin{remark} We assume that sample splitting is employed, where separate data folds are used for estimation of the nuisance functions and for computing the empirical average $\mathbb{Q}_N$ in the final  estimators. This cross-fitting approach controls empirical process complexity without additional model assumptions \citep{ehk_semiparametric}.
\end{remark}

This approximation to an empirical average of EIF values implies that our estimator is approximately Gaussian, enabling valid confidence interval construction using standard methods. The following theorem formalizes this result.

\begin{theorem}
   Under the conditions of \Cref{thm:optimal}, the one-step ML estimator $\widehat{\psi}^{-1}_{\text{os}}$ is approximately normally distributed. Consequently, the following is an approximately $1-\alpha$ confidence interval for the inverse capture probability $\psi^{-1}$:
\begin{equation}
\label{eq:ci_psi}
    \widehat{CI}(\psi^{-1}) = \left[\widehat{\psi}^{-1}_{os} \pm z_{1-\alpha/2} \widehat{\sigma}/\sqrt{N}\right],
\end{equation}
where $\widehat{\sigma}^2 = \widehat{var}(\widehat{\phi})$ is the unbiased empirical variance of the estimated EIF.
\label{thm:ci_psi}
\end{theorem}

In summary, the one-step estimator enjoys three critical properties that hold for any finite sample size: (i) it approximates the empirical average of the EIF with high probability, with deviations controlled by the second-order remainder term, (ii) it is approximately Gaussian enabling standard inference procedures, and (iii) it exhibits a doubly robust structure where the second-order remainder depends on products of nuisance function estimation errors, providing robustness to misspecification. These properties ensure that the estimator is minimax optimally efficient while maintaining validity under flexible machine learning estimation of the nuisance $q$-probabilities and missingness propensity score.

\begin{remark}
When there is no missingness ($R_i = 1$ for all $i$), our proposed one-step estimator reduces to the estimator developed in \cite{dulce2024population} for log-linear models under the NHOI assumption, including their ``log-linear model among the unobserved'' specification and the two list case in \cite{das2024doubly}. The present work thus extends their estimation frameworks to accommodate missing covariate information through the MAR assumption.
\end{remark}

\subsection{Population Size Estimation and Inference}

Having developed estimators for the inverse capture probability $\psi^{-1}$, we now construct the corresponding population size estimator and derive its confidence interval. The population size $n$ can be estimated via the plug-in relationship:
\begin{equation}
\label{eq:n_estimator}
    \widehat{n}_{\text{os}} = \frac{N}{\widehat{\psi}_{\text{os}}}.
\end{equation}

For inference on the population size $n$, we leverage \Cref{thm:ci_psi} for the (linear) transformation $n = N/\psi$. The variance of $\widehat{n}_{\text{os}}$ incorporates both the uncertainty in estimating $\psi$ (captured through the EIF variance) and the binomial variation in observing $N$ individuals:
\begin{equation}
\label{eq:ci_n_inference}
    \widehat{CI}(n) = \left[ \widehat{n}_{\text{os}} \pm z_{1-\alpha/2}\sqrt{\widehat{n}_{\text{os}}\left(\widehat{\psi}_{\text{os}}\widehat{\sigma}^2 + \frac{1-\widehat{\psi}_{\text{os}}}{\widehat{\psi}_{\text{os}}} \right)}\right],
\end{equation}
where $\widehat{\sigma}^2 = \widehat{\text{var}}(\widehat{\phi})$ is the empirical variance of the estimated EIF. This confidence interval maintains approximately valid coverage for any finite sample size $N$, with coverage controlled by the second-order remainder term from \Cref{thm:optimal}. This result can be seen as a direct application of  Theorem 4 in \cite{das2024doubly}, provided that $\widehat{\psi}_{os}^{-1}$ is well-approximated by an empirical average.

\section{Empirical Data Analysis}
\label{sec:experiments}

We demonstrate the practical utility of our nonparametric methodology through simulation studies evaluating finite-sample performance and for estimating mortality in the Gaza conflict using capture-recapture data from \cite{jamaluddine2025lancet}. Both applications highlight the importance of properly accounting for missing covariate information in population size estimation tasks. Code implementing our methods to reproduce all results is available at \url{https://github.com/mdulcer/population-size-missing-covariates}.

\subsection{Synthetic Capture-Recapture Data}

To evaluate the finite-sample performance of our proposed methodology, we conduct simulation studies under realistic capture-recapture settings where the true population size is known. We compare our doubly robust ML and plug-in estimators under the MAR assumption against a common approach in the literature that imputes missing covariates before estimation, examining performance as missingness rates increase from 0\% to 70\%.

\paragraph*{Data generation.} 
We fix a total population size $n=$5,000 and generate the following always-observed variables $V$ for each unit $i=1,\dots,n$:
\begin{align*}
    sex & \sim Bernoulli(p = 0.48),\\ 
    age & \sim  \mathcal{N}(\mu = 40,~ \sigma^2 = 144),
\end{align*}
where \textit{age} is truncated between 10 and 70 years. In addition, we generate the potentially missing covariate 
$$X \sim Uniform(\{1, \dots, 10\}),$$
representing an unordered categorical variable with 10 levels, that can be interpreted as a geographical region or time period.

We generate heterogeneous capture probabilities for two capture occasions under conditional independence (in this case equivalent to \Cref{ass:nhoi} - NHOI):
\begin{align*}
    \PP(Y_1 = 1 \mid V, X) & = \text{expit}(a - 0.5\cdot sex + 0.005\cdot (age-45)^{2} + (2\cdot sex - 1)\cdot \beta_X) \\
    \PP(Y_2 = 1 \mid V, X) & = \text{expit}(a + 0.5\cdot sex + 0.005\cdot (age-35)^{2} + (sex - 0.5)\cdot \beta_X),
\end{align*}
where $\beta_X = [-1, -0.6, -0.2, 0.2, 0.6, 0.0, -0.4, 0.4, 0.8, 1]$ represent category-specific effects, and $a$ is calibrated such that the marginal capture probability $\psi = 0.7$. This yields approximately $N \approx$ 3,500 observed individuals with $Y_1 + Y_2 > 0$ in each simulation. 

We simulate a missingness mechanism satisfying MAR (\Cref{ass:mar}), where missingness depends on the capture profiles $y$ and \textit{sex}, but not on the unobserved $X$ value:
\begin{align*}
    \PP(R=0 \mid y = (0,1), ~ sex = 1) & = 5r \\
    \PP(R=0 \mid y = (0,1), ~ sex = 0) & = r \\
    \PP(R=0 \mid y = (1,0)) & = 3r \\
    \PP(R=0 \mid y = (1,1)) & = r,
\end{align*}
where the base rate $r$ is calibrated to achieve marginal missingness probabilities $\PP(R=0)$ ranging from 0\% to 70\%.

\paragraph*{Estimation procedure.} We compare four approaches that vary along two dimensions: one-step ML and plug-in estimators, each computed under either MAR-based handling or \textit{a priori} imputation of missing data. Specifically, the imputation approaches first impute missing $X$ values using a Random Forest model  based on the always-observed covariates $v$ and capture profiles $y$, and then compute the proposed plug-in and one-step ML estimators treating imputed data as complete. All approaches estimate nuisance functions $q_y(v), \lambda_x(y,v),$ and $\pi(y,v)$ using Random Forest models and compute final estimates using 2-fold cross-fitting. Note that while we use Random Forest to illustrate the nonparametric estimation step, our methodology is agnostic to the choice of nuisance function estimation method.

We evaluate performance across 50 simulation replicates using three metrics: average absolute bias $\frac{1}{50}\sum_{i=1}^{50}|\widehat{n}_i - n|$, root mean squared error (RMSE) $\sqrt{\frac{1}{50}\sum_{i=1}^{50}(\widehat{n}_i - n)^2}$, and nominal 95\% confidence interval coverage $\frac{1}{50}\sum_{i=1}^{50}\mathbbm{1}(n \in \widehat{CI}_i(n))$. Note that plug-in estimators do not have well-defined variance formulas in our framework, so we use the estimated variance from the corresponding one-step estimator to construct confidence intervals for the plug-in approaches.

\paragraph*{Results.}
Figure \ref{fig:sim_bias_rmse} presents average absolute bias and RMSE as missingness increases. The MAR-based estimators maintain low bias  ($\approx$ 100 deaths) and RMSE across all missingness rates, with minimal differences between the one-step ML and plug-in estimators. In sharp contrast, imputation-based approaches exhibits rapidly increasing bias and RMSE, with the one-step with imputation exceeding 800 bias ($\approx$ 20\% error) at high missingness rates. This suggests that the overall estimation error is dominated by bias introduced in the imputation step rather than variance, where the one-step estimator's bias correction mechanism, designed for MAR settings, can amplify rather than reduce bias when applied to naively imputed data.

\begin{figure}[htpb]
\centering
\begin{minipage}{0.48\textwidth}
\centering
\includegraphics[width = \textwidth]{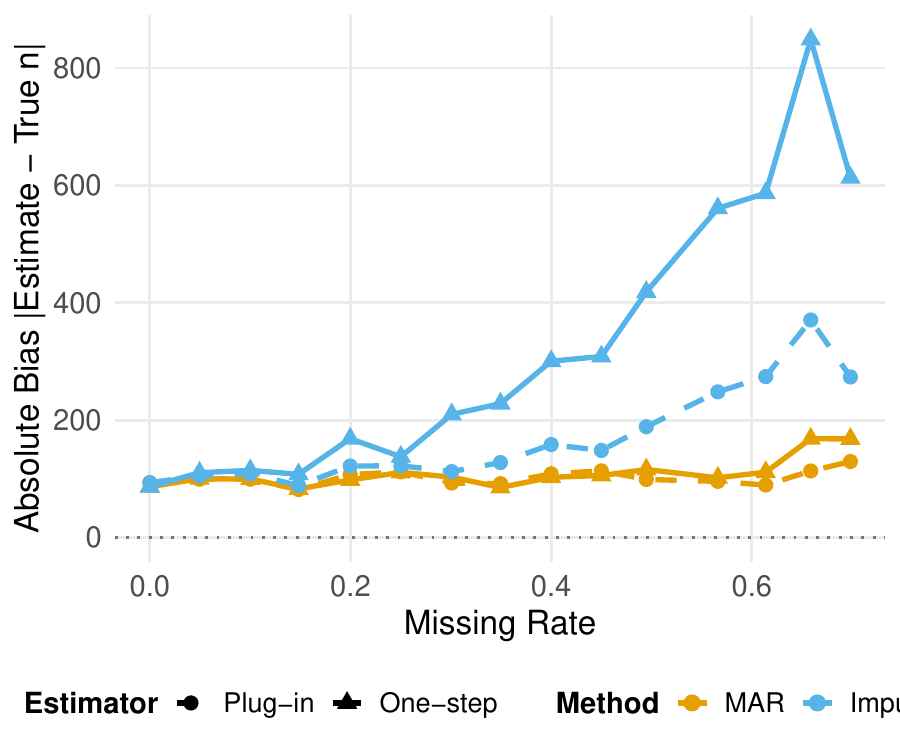}
\end{minipage}
\hfill
\begin{minipage}{0.48\textwidth}
\centering
\includegraphics[width = \textwidth]{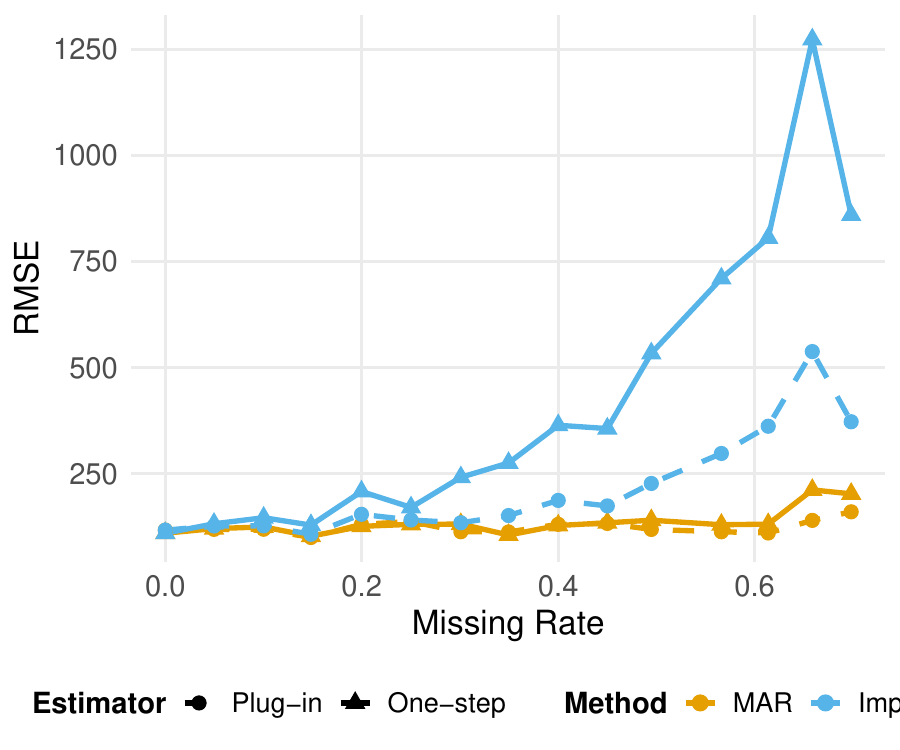}
\end{minipage}
\caption{Average absolute bias (left) and RMSE (right) across missingness rates. MAR-based estimators (orange lines) maintain consistently good performance, while naive imputation approaches (blue lines) show rapidly increasing bias and RMSE as missingness increases. Imputation-based one-step estimator (solid) exhibits particularly poor performance.}
\label{fig:sim_bias_rmse}
\end{figure}

Figure \ref{fig:sim_coverage} reveals a similar pattern in inference validity. Both MAR-based methods maintain approximately nominal coverage across all missingness scenarios, with the one-step ML estimator showing slightly better performance. On the other hand, both imputation-based approaches exhibit severe coverage reduction at high missingness rates, dropping below 75\% at 25\% missingness. That is, standard confidence intervals rapidly become invalid for inference in moderate to large missingness regimes handled using a priori imputation. Note that plug-in estimators lack well-defined variance formulas, so their coverage rates use variance estimates from the corresponding one-step estimator and should be interpreted as approximate.

\begin{figure}[htpb]
\centering
\includegraphics[width = 0.58\textwidth]{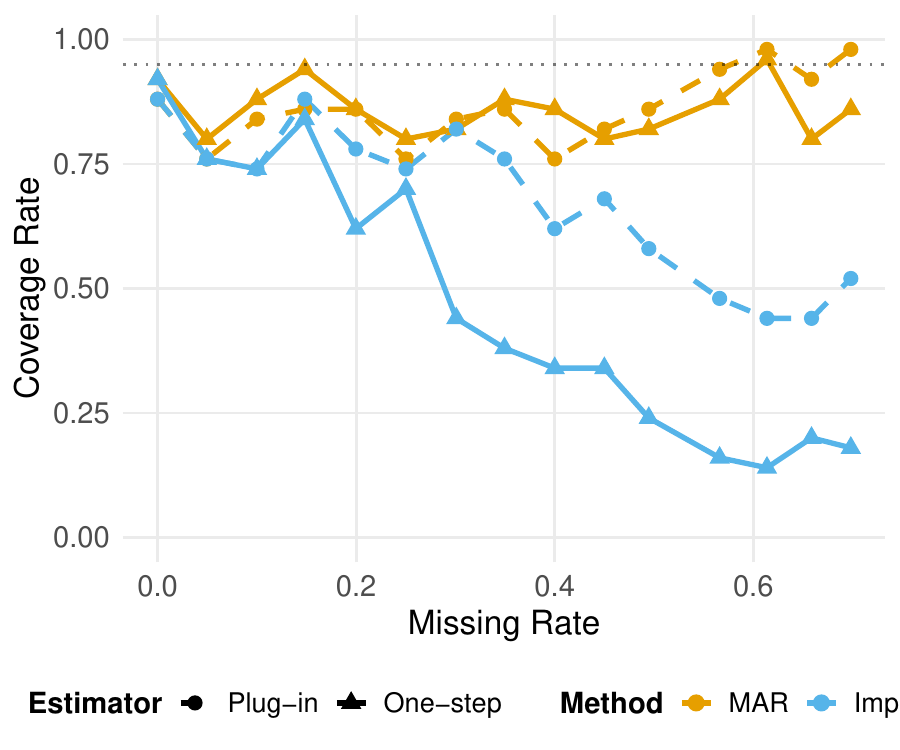}
\caption{Nominal 95\% confidence interval coverage rates. MAR-based methods (orange lines) maintain approximately valid inference across all missingness levels. Imputation-based approaches (blue lines) exhibit severe coverage reduction, with the one-step variant showing particularly poor performance. The dotted reference line indicates nominal 95\% coverage.}
\label{fig:sim_coverage}
\end{figure}

These simulation results demonstrate that properly accounting for missingness through principled frameworks is essential for valid capture-recapture inference. Naive imputation can introduce substantial bias that rapidly increases with more predominant missingness and renders confidence intervals invalid. These findings validate our theoretical contributions and underscore the practical necessity of our proposed methodology for real-world population size estimation with missing covariate data.

\subsection{Estimation of the Mortality in the Gaza Strip from Oct 7, 2023, to June 30, 2024}
\label{sec:gaza}

A recent study by \citet{jamaluddine2025lancet} estimated Gaza mortality from October 7, 2023, to June 30, 2024, using three-list capture-recapture data. Their approach: (1) imputed missing covariates on sex, age and month of death using Multiple Imputation by Chained Equations \citep{white2011multiple}, (2) fitted multiple Poisson log-linear models under the no highest-order interaction assumption with different list dependencies, and (3) averaged estimates across models. The authors reported an estimated 64,260 total deaths (95\% CI: (55,298-78,525)) due to traumatic injury during the study period, implying 41\% underreporting relative to the 37,877 official death count reported by the Palestinian Ministry of Health (MoH) for the same period.\footnote{The 37,877 figure reported by MoH includes both identified and unidentified decedents \citep{jamaluddine2025lancet}. However, to our knowledge, the specific data collection and aggregation methodologies used to construct this count are not publicly available.}

However, \citeauthor{jamaluddine2025lancet}'s approach faces two compounding issues: first, different list dependency structures represent distinct identification strategies, of which at most one can be correct; second, even if the correct dependency structure were included, the parametric log-linear models add potential misspecification bias. In contrast, our approach uses a single, weaker identification assumption (NHOI), nonparametric estimation, and directly accommodates missing covariate information under a Missing at Random assumption without requiring separate imputation or model averaging steps. In the following, we re-analyze this capture-recapture data using our  doubly robust ML framework.

\paragraph*{Data.}
We use the de-duplicated and record-matched data recording 29,271 unique confirmed deaths due to traumatic injury in the Gaza Strip between October 7, 2023, and June 30, 2024. The data consists of the following three lists:\footnote{The data is publicly available here: \url{https://github.com/ZeinaJamaluddine/gaza_mortality_capture_recapture}.}

\begin{itemize}
    \item Hospital list: Palestinian Ministry of Health hospitals recorded 22,347 identified decedents with Palestinian ID numbers, names, month and location of death, age, and sex.
    \item Survey list: Gaza MoH online survey documented 7,581 deaths with Palestinian ID numbers, names, age, sex, location, and reporting source.
    \item Social media list: 3,190   online obituaries from platforms including ``Gaza Shaheed'', ``Martyrs of Gaza'', and ``The Palestinian Information Center'', typically containing names, age, and date/location of death, originally collected by \cite{jamaluddine2025lancet}.
\end{itemize}

Covariate information was complete for 96.64\% of records (28,289 individuals). Among those with complete age and sex data, 37.2\% were female and 62.8\% male. Children under 18 accounted for 33.3\% of deaths, while adults aged 65+ represented 5.8\%. \Cref{tbl:missingness} reveals strong dependence between missingness and capture profiles: missing information predominantly occurs among individuals captured only by the social media list (L001), with 979 of 1,469 such records (66.6\%) having incomplete age or month information, compared to essentially complete data for other capture profiles.
\begin{table}[ht]
\centering
\begin{tabular}{lrrrrrrr}
\toprule
  & \multicolumn{7}{c}{Capture Profile} \\
Complete cases & L001 & L010 & L011 & L100 & L101 & L110 & L111 \\
\midrule
FALSE & 979 & 3 & 0 & 0 & 0 & 0 & 0 \\
TRUE & 490 & 5101 & 351 & 19048 & 1173 & 1929 & 197 \\
\bottomrule
\end{tabular}
\caption{Missingness patterns across capture profiles. Rows indicate complete (TRUE) or incomplete (FALSE) information for age and month of death. Columns show capture profiles where L100 = hospital only, L010 = survey only, L001 = social media only, and intersections represent multiple list captures.}
\label{tbl:missingness}
\end{table}

\paragraph*{Estimation procedure.}
We specify sex as the always-observed covariate ($V$), while age (discretized as 0-14, 15-29, 30-44, 45-59, $\geq$60 years) and month of death are used as the potentially missing covariates ($X$).\footnote{56 records (0.19\%) had missing sex information, which we imputed using Random Forest prior to analysis. Results are robust to excluding these observations.} Our Missing at Random assumption states that missingness depends on sex and capture profile but not on the unobserved age or month values themselves, a plausible assumption given the missingness concentration in the social media list, where obituary completeness likely depends on reporting practices rather than victim characteristics. Finally, to illustrate our flexible methodology, we estimate the nuisance functions $q_y(v)$, $\lambda_x(y,v)$, and $\pi(y,v)$ using Random Forest with 5-fold cross-fitting.

\paragraph*{Results.}
\Cref{tab:gaza_results} presents mortality estimates across methodologies. Our nonparametric plug-in estimator (65,294 deaths; 95\% CI: (56,557-74,030)) closely aligns with the original study's estimate. However, our doubly robust one-step ML estimator yields a more conservative estimate of 59,441 deaths (95\% CI: (50,708-68,173)), approximately 10\% lower than previous methods.

\begin{table}[ht]
\centering
\begin{tabular}{lcc}
\toprule
Method & Mortality estimate & 95\% CI \\
\midrule
Observed data & 29,271 & — \\
\citet{jamaluddine2025lancet} & 64,260 & (55,298 - 78,525) \\
Plug-in ML & 65,294 & (56,557 - 74,030) \\
One-step ML & 59,441 & (50,708 - 68,173) \\
\bottomrule
\end{tabular}
\caption{Gaza mortality estimates between October 7, 2023 and June 30, 2024. Confidence interval for the plug-in estimator uses the variance estimates from the corresponding one-step estimator and should be interpreted as approximate.}
\label{tab:gaza_results}
\end{table}

This reduction likely stems from three methodological improvements: (1) direct accommodation of missing data through augmented inverse probability weighting rather than separate imputation, (2) flexible nonparametric estimation of capture heterogeneity without parametric model specification, and (3) bias correction through the efficient influence function. Notably, while the confidence intervals from the different methods overlap, the one-step estimator produces a narrower interval, suggesting more precise estimation due to its approximately optimal efficiency properties. The doubly robust estimator's finite-sample optimality guarantees indicate it provides more reliable inference, particularly given the moderate sample size and complex missingness patterns in this conflict setting.

Critically, all capture-recapture methods indicate substantial underreporting: even our most conservative estimate suggests approximately one in four conflict deaths were not captured in the Palestinian Ministry of Health's official count of 37,877 deaths during the study period. This underscores both the humanitarian crisis and the methodological importance of robust statistical approaches for conflict casualty estimation.
\section{Discussion}

In this work we developed a doubly robust machine learning framework for population size estimation in capture-recapture settings under missing covariate information. Our approach combines a no highest-order interaction assumption in log-linear models with a missing at random assumption to achieve identification, both broadly applicable in realistic capture-recapture settings. We then leverage nonparametric efficiency theory to construct one-step estimators that maintain validity under flexible machine learning estimation of nuisance functions. The methodology provides finite-sample concentration guarantees and approximately valid inference for any sample size in a doubly robust fashion, addressing a critical gap in existing capture-recapture methods that either ignore missing data or handle it through \textit{a priori} imputation. Applications to Gaza conflict mortality estimation and simulation studies demonstrate substantial improvements over direct imputation approaches, particularly as missingness rates increase.

Our framework is particularly relevant given that record linkage, a necessary preprocessing step in most capture-recapture applications, frequently generates missing covariate patterns. When perfect unique identifiers are unavailable, prevalent in hard-to-reach populations, linking records across capture occasions may produce conflicting covariate values requiring harmonization. Common \textit{canonicalization} techniques randomly select one value according to human expert judgment \citep{binette2022almost}, potentially introducing bias into subsequent analyses and creating complex missingness patterns in the final linked dataset. While our estimation framework assumes perfect record linkage, as do most population size estimation methods, it provides a principled approach to handle such missing data under MAR with robust statistical guarantees, avoiding the need for arbitrary canonicalization decisions. Extending our framework to jointly model record linkage uncertainty and missing covariates, for instance through linkage-averaging \citep{sadinle2018bayesian}, represents an important direction for future research.

Another promising methodological research avenue is to develop individual estimators for the full conditional capture probability $q_y(V,X)$ and the conditional expectation $\mathbb{E}[\gamma(V,X)^{-1} \mid V]$ under missing covariates. Crucially, $q_y(V,X)$ plays a critical role in the second-order remainder term $\widehat{R}_2$, influencing the finite-sample properties of our estimator $\widehat{\psi}_{os}^{-1}$. In the present work, we estimate these quantities using plug-in approaches derived from the MAR identification: $\widehat{q}_y(v, x) = \widehat{\lambda}_x(y,v)\widehat{q}_y(v)/\sum_{y'\neq0}\widehat{\lambda}_x(y',v)\widehat{q}_{y'}(v)$ and  $\widehat{\mathbb{E}}[\widehat{\gamma}(V,X)^{-1} \mid V=v] = \sum_x \widehat{\gamma}^{-1}(v,x)\widehat{q}(x|v)$, respectively. Developing one-step estimators, especially for $q_y(V,X)$, and incorporating these into $\widehat{\psi}^{-1}_{os}$ could potentially further reduce the remainder term and improve finite-sample efficiency, particularly under complex covariate structures or high missingness rates. For the conditional expectation $\mathbb{E}[\gamma(V,X)^{-1} \mid V]$, more sophisticated estimation strategies become valuable when $X$ is high-dimensional or continuous. For instance, tracking population changes across geographic regions or temporal periods requires estimating these conditional expectations flexibly and efficiently. 

Our work contributes to the growing literature on nonparametric methods for total population size estimation. As capture-recapture data become increasingly available through linked administrative sources, electronic health records, and social media, principled approaches to handling data quality issues such as missing information, measurement error, and linkage uncertainty become urgent. Our doubly robust ML framework demonstrates that modern semiparametric theory can leverage machine learning methods to address these challenges while maintaining statistical rigor and providing practical inference guarantees.

\bibliographystyle{apalike}
\bibliography{bib}

\newpage
\section*{Appendix}

\begin{proof}[Proof of \Cref{lemma:if_psi}]

We show that the proposed function 
\begin{align*}
    \phi(Z) =  \sum_{y\neq 0} (-1)^{1+|y|} & \left\{ \frac{R\mathbbm{1}(Y=y)}{\pi(y,v)}\left(\frac{1}{q_y(v,x)}\left(\frac{1}{\gamma(v,x)}-1\right) - \frac{1}{q_y(v)}\mathbb{E}\left[\frac{1}{\gamma(v,X)} - 1 ~\Bigg|~ V=v\right]\right) \right.  \\ 
    & \left. +\left(\frac{\mathbbm{1}(Y=y)}{q_y(v)}\right)\mathbb{E}\left[\frac{1}{\gamma(v, X)} - 1 ~\Bigg|~ V=v\right]  \right\}  - \left(\frac{1}{\psi} - 1\right),
\end{align*}
is the Efficient Influence Function for $\psi^{-1} = \E[\gamma(V,X)^{-1}]$, by deriving the remainder term $\widehat{R}_2$ in its von Mises expansion and showing it is of second-order. Let $\bar{\psi} = \psi(\bar{\QQ})$ be the plug-in estimator of the capture probability for a generic distribution $\bar{\QQ}$. Then:
\begin{align*}
      R_2(\QQ, &\bar{\QQ}) =  \E_{\QQ}\left[\phi(Z; \bar{\QQ})\right] + \frac{1}{\bar{\psi}} - \frac{1}{\psi} \\
    =  \E_{\QQ}&\left[\sum_{y\neq 0} (-1)^{1+|y|}  \left\{ \frac{R}{\bar{\pi}(y,V)}\left(\frac{\mathbbm{1}(Y=y)}{\bar{q}_y(V,X)}\left(\frac{1}{\bar{\gamma}(V,X)}-1\right) - \frac{\mathbbm{1}(Y=y)}{\bar{q}_y(V)}\bar{\mathbb{E}}\left[\frac{1}{\bar{\gamma}(V,X)} - 1 ~\Bigg|~ V\right]\right) \right.\right. \\ 
     + & \left.\left.\left(\frac{\mathbbm{1}(Y=y)}{\bar{q}_y(V)}\right)\bar{\mathbb{E}}\left[\frac{1}{\bar{\gamma}(V, X)} - 1 ~\Bigg|~ V\right]  \right\}  - \left(\frac{1}{\bar{\psi}} - 1\right)\right] + \frac{1}{\bar{\psi}} - \frac{1}{\psi} \\ 
     =  \E_{\QQ}&\left[\sum_{y\neq 0} (-1)^{1+|y|}  \left\{ \frac{\pi(y,V)}{\bar{\pi}(y,V)}\left(\frac{q_y(V,X)}{\bar{q}_y(V,X)}\left(\frac{1}{\bar{\gamma}(V,X)}-1\right) - \frac{q_y(V)}{\bar{q}_y(V)}\left(\frac{1}{\bar{\gamma}(V,X)}-1\right)\right) \right.\right. \\ 
     + & \left.\left.\left(\frac{q_y(V)}{\bar{q}_y(V)}\right)\left(\frac{1}{\bar{\gamma}(V,X)}-1\right)  \right\}  - \left(\frac{1}{\gamma(V,X)} - 1\right)\right]  \\ 
     = \E_{\QQ}&\left[  \left(\frac{1}{\bar{\gamma}(V,X)}-1\right)\sum_{y\neq 0} (-1)^{1+|y|}\left(\frac{\pi(y,V)}{\bar{\pi}(y, V)}-1\right)\left(\frac{q_y(V,X)}{\bar{q}_y(V,X)} - \frac{q_y(V)}{\bar{q}_y(V)}\right) \right.\\
     + & \left.\left(\frac{1}{\bar{\gamma}(V,X)}-1\right)\sum_{y\neq 0} (-1)^{1+|y|} \frac{q_y(V,X)}{\bar{q}_y(V,X)} - \left(\frac{1}{\gamma(V,X)} - 1\right)\right] \\
     = \E_{\QQ}&\left[  \left(\frac{1}{\bar{\gamma}(V,X)}-1\right)\sum_{y\neq 0} (-1)^{1+|y|}\left(\frac{\pi(y,V) - \bar{\pi}(y,V)}{\bar{\pi}(y, V)}\right)\left(\frac{q_y(V,X) - \bar{q}_y(V,X)}{\bar{q}_y(V,X)} - \frac{q_y(V) - \bar{q}_y(V)}{\bar{q}_y(V)}\right) \right.\\
     & + \left(\frac{1}{\widetilde{\gamma}(V,X)}-1\right)\sum_{y_i\neq y_j \neq 0} (-1)^{|y_i| + |y_j|}\frac{(q_{y_i}(V,X) - \bar{q}_{y_i}(V,X))}{\widetilde{q}_{y_i}(V,X)}\frac{(q_{y_j}(V,X) - \bar{q}_{y_j}(V,X))}{\widetilde{q}_{y_j}(V,X)} \nonumber \\ 
    & + \left. \left(\frac{1}{\widetilde{\gamma}(V,X)}-1\right)\sum_{|y|\neq 0, ~even} \left(\frac{q_{y}(V,X) - \bar{q}_{y}(V,X)}{\widetilde{q}_{y}(V,X)}\right)^2\right],
\end{align*}
for some intermediate $\widetilde{q}(V,X)$ between $q(V,X)$ and $\widehat{q}(V,X)$. The last equality follows from the second-order Taylor expansion for
$$\frac{1}{\gamma(X)} - 1 = \exp\left(\sum_{y\neq 0} (-1)^{1+|y|}\log(q_y(V,X))\right)$$ 
around $\{\bar{q}_y(V,X)\}_{y\neq 0}$, following the strategy in \cite{dulce2024population}.

To conclude our proof, note that our nonparametric model does not restrict the tangent space \citep{semiparametric, bickel1993efficient}, which implies that $\phi(Z)$ is the (unique) EIF for $\psi^{-1}$ by Lemma 2 in \cite{kennedy2023density}.

\end{proof}

\begin{proof}[Proof of \Cref{thm:optimal}]

The result is a direct application of Markov's inequality given an error tolerance $\eta > 0$:
\begin{align*}
    \PP\left[\Bigg|\frac{1}{\widehat{\psi}_{os}} - \frac{1}{\psi} - \QQ_{N}(\phi)\Bigg| \leq \eta \right] & \geq 1- \frac{1}{\eta^2}\E\left[\Bigg|\frac{1}{\widehat{\psi}_{os}} - \frac{1}{\psi} - \QQ_{N}(\phi)\Bigg|^2\right] \\
    & = 1 - \frac{1}{\eta}\E\left[\Big|\widehat{R}_2 + (\QQ_{N} - \QQ)(\widehat{\phi} - \phi))\Big|^2\right] \\ 
    & \geq 1 - \frac{1}{\eta^2}\E\left[\widehat{R}_2^2 + \frac{1}{N}\norm{\widehat{\phi} - \phi}^2\right].
\end{align*}
The second equality follows from the definition of our one-step estimator $$\widehat{\psi}^{-1}_{os} = \widehat{\psi}^{-1}_{pi} + \QQ_{N}(\widehat{\phi}),$$ 
and the von Mises expansion \citep{bickel1993efficient, ehk_semiparametric}
$$\widehat{\psi}^{-1}_{pi} - \psi^{-1} = -\QQ(\widehat{\phi}) + \widehat{R}_2,$$
where $\widehat{R}_2$ was derived in the proof of \Cref{lemma:if_psi}. Finally, note that the EIF has mean zero ($\QQ(\phi) = 0$), concluding our proof.

\end{proof}

\begin{proof}[Proof of \Cref{thm:ci_psi}]

From Theorem \ref{thm:optimal}, we have that
$$\widehat{\psi}^{-1}_{os} - \psi^{-1} = \mathbb{Q}_N(\phi) + \widehat{R}_2 + o_P(N^{-1/2}),$$
where the deviation from $\mathbb{Q}_N(\phi)$ is controlled by the second-order remainder $\widehat{R}_2$.

Since $\mathbb{Q}_N(\phi) = \frac{1}{N}\sum_{i=1}^N \phi(Z_i)$ is an empirical average of iid random variables $Z_i$ with $\mathbb{E}[\phi(Z)] = 0$ (by construction of the EIF) and finite variance $\text{Var}(\phi) = \sigma^2$, the Central Limit Theorem implies
$$\frac{\widehat{\psi}^{-1}_{os} - \psi^{-1}}{\widehat{\sigma}/\sqrt{N}} \approx \mathcal{N}(0,1),$$
where $\widehat{\sigma}^2$ is the empirical variance of $\widehat{\phi}$.

Therefore,
$$\mathbb{P}\left(\psi^{-1} \in \left[\widehat{\psi}^{-1}_{os} \pm z_{1-\alpha/2} \widehat{\sigma}/\sqrt{N}\right]\right) \approx 1 - \alpha,$$
where the coverage error depends on the magnitude of the second-order remainder $\widehat{R}_2$.

\end{proof}

\end{document}